\documentclass[envcountsame]{llncs}

\usepackage{amssymb}
\usepackage{graphicx}
\usepackage{booktabs}
\usepackage{pdflscape}
\usepackage{mdframed}
\usepackage{subfig}
\usepackage{mathtools}

\let\doendproof\endproof
\renewcommand\endproof{~\hfill\qed\doendproof}

\DeclareMathOperator{\tw}{tw}
\newcommand{\SB}{\{\,} \newcommand{\SM}{\;{|}\;} \newcommand{\SE}{\,\}}

\renewcommand{\P}{\textsc{P}}
\newcommand{\NP}{\textsc{NP}}
\newcommand{\coNP}{\textsc{coNP}}
\newcommand{\containment}{\ensuremath{\mathsf{\NP\subseteq \coNP/poly}}}

\usepackage{xspace}
\usepackage{framed}

\newcommand{\ProblemFormat}[1]{{\sc #1}}
\newcommand{\ProblemName}[1]{\ProblemFormat{#1}\xspace}

\DeclareMathOperator{\fall}{Fall}

\newcommand{\probThreeSAT}{\ProblemName{$3$-SAT}}
\newcommand{\probkCol}{\ProblemName{$k$-Coloring}}
\newcommand{\probThreeCol}{\ProblemName{$3$-Coloring}}
\newcommand{\probChromIndex}{\ProblemName{Edge $k$-Coloring}}

\newcommand{\probkFallCol}{\ProblemName{Fall $k$-Coloring}}
\newcommand{\probThreeFallCol}{\ProblemName{Fall $3$-Coloring}}

\newcommand{\probPlanarMonotoneThreeSat}{\ProblemName{Planar Monotone $3$-SAT}}

\newcommand{\probSetPart}{\ProblemName{Set Partition}}

\begin{document}
\title{Complexity of fall coloring for\\restricted graph classes}
%
%
\author{Juho Lauri\inst{1} \and Christodoulos Mitillos\inst{2}}
\authorrunning{J. Lauri and C. Mitillos}
%
\institute{Nokia Bell Labs, Dublin, Ireland\\\email{juho.lauri@gmail.com}\\
\and
Illinois Institute of Technology, Chicago, United States\\\email{cmitillo@iit.edu}}

\maketitle

\begin{abstract}
We strengthen a result by Laskar and Lyle (Discrete Appl.\ Math.\ (2009), 330--338) by proving that it is $\NP$-complete to decide whether a bipartite planar graph can be partitioned into three independent dominating sets.
In contrast, we show that this is always possible for every maximal outerplanar graph with at least three vertices.
Moreover, we extend their previous result by proving that deciding whether a bipartite graph can be partitioned into $k$ independent dominating sets is $\NP$-complete for every $k \geq 3$.
We also strengthen a result by Henning~{et al.}~(Discrete Math.\ (2009), 6451--6458) by showing that it is $\NP$-complete to determine if a graph has two disjoint independent dominating sets, even when the problem is restricted to triangle-free planar graphs.
Finally, for every $k \geq 3$, we show that there is some constant $t$ depending only on $k$ such that deciding whether a $k$-regular graph can be partitioned into $t$ independent dominating sets is $\NP$-complete.
We conclude by deriving moderately exponential-time algorithms for the problem.

\keywords{Fall coloring  \and Independent domination \and Computational complexity.}
\end{abstract}

\section{Introduction}
Domination and independence are two of the most fundamental and heavily-studied concepts in graph theory.
In particular, a partition of the vertices of a graph into independent sets is known as \emph{graph coloring} --- a central problem with several practical applications in e.g., scheduling~\cite{Marx2004}, timetabling, and seat planning~\cite{Lewis2015}.
In addition, independence and domination are central to various problems in telecommunications, such as adaptive clustering in distributed wireless networks and various channel assignment type problems such as code assignment, frequency assignment, and time-slot assignment.
For an overview, see~\cite[Chapter~30]{Resende2008}.

Let $G=(V,E)$ be a graph and let $S \subseteq V$ be a subset of its vertices.
Here, $S$ is an \emph{independent set} if the vertices in $S$ are pairwise non-adjacent.
We say that $S$ is a \emph{dominating set} when every vertex of $V$ either is in $S$ or is adjacent to a vertex in $S$.
Combining these properties, Dunbar~{et al.}~\cite{Dunbar2000} studied the problem of partitioning the vertex set of a graph into sets that are both independent and dominating.
The authors viewed this problem as a kind of a graph coloring defined as follows.
Let $\Pi = \{V_1,V_2,\ldots,V_k\}$ be a partition of $V$.
We say that a vertex $v \in V_i$ for $i \in \{1,2\ldots,k\}$ is \emph{colorful} if $v$ is adjacent to at least one vertex in each color class $V_j$ for $i \neq j$.
$\Pi$ is a \emph{fall $k$-coloring} if each $V_i$ is independent and every vertex $v \in V$ is colorful.
Informally, in a fall coloring, each vertex has in its immediate neighborhood each of the colors except for its own.
For an illustration of the concept, see Figure~\ref{fig:fk345}.
For possible applications of fall $k$-coloring, including transceiver frequency allocation and timetabling, see~\cite[Section~4.2]{Mitillos2016}.

The maximum $k$ for which a graph $G$ has a fall $k$-coloring is known as the \emph{fall achromatic number}, denoted by $\psi_{fall}(G)$.
Clearly, we have $\psi_{fall}(G) \leq \delta(G) + 1$, where $\delta(G)$ is the minimum degree of $G$.
Similarly, the minimum $k$ for which a graph $G$ has a fall $k$-coloring is known as the \emph{fall chromatic number}, denoted by $\chi_{fall}(G)$.
Here, it holds that $\chi(G) \leq \chi_{fall}(G)$, where $\chi(G)$ is the chromatic number of $G$ (see~\cite{Dunbar2000}).
The \emph{fall set} of a graph $G$, denoted by $\fall(G)$, is the set of integers $k$ such that $G$ admits a fall $k$-coloring.
In general, $\fall(G)$ is not guaranteed to be non-empty, but it is finite for finite graphs.
For example, we have that $\fall(C_6) = \{ \chi_{fall}(C_6), \psi_{fall}(C_6) \} = \{ 2,3 \}$ with a fall 3-coloring shown in Figure~\ref{fig:fk345}.
To obtain a fall 2-coloring for $C_6$, it suffices to observe that any 2-coloring of a connected bipartite graph is a fall 2-coloring. 

In this work, our focus is on the computational complexity of fall coloring.
In this context, it was shown by Heggernes and Telle~\cite{Heggernes1998} that for every $k \geq 3$, it is $\NP$-complete to decide whether a given graph $G$ has $k \in \fall(G)$.
Laskar and Lyle~\cite{Laskar2009} improved on this in the case of $k=3$ by showing that it is $\NP$-complete to decide whether a given bipartite graph $H$ has $3 \in \fall(H)$.
On a positive side, it was shown by Telle and Proskurowski~\cite{Telle1997} that deciding whether $k \in \fall(G)$ can be done in polynomial time when $G$ has bounded cliquewidth (or  treewidth).
This can also be derived from the fact that the property of being fall $k$-colorable can be expressed in monadic second order logic (for details, see~\cite[Section~7.4]{fpt-book}).
For chordal graphs $G$, it is known that the fall set is either empty or contains exactly $\delta(G)+1$.
To the best of our knowledge, the complexity of deciding this case is open.
For subclasses of chordal graphs, the fall sets of threshold and split graphs can be characterized in polynomial time~\cite{Mitillos2016}.
Despite independence and domination being central concepts in graph theory, we are unaware of any further hardness results for fall coloring (see also e.g.,~\cite[Section~7]{Goddard2013}).

We extend and strengthen previous hardness results for fall coloring, and provide new results as follows:
\begin{itemize}
\item In Section~\ref{sec:planar-bip}, we extend the result of Laskar and Lyle~\cite{Laskar2009} by proving that for $k \geq 3$, it is $\NP$-complete to decide whether a bipartite graph is fall $k$-colorable.
Further, for the case of $k = 3$, we strengthen their result considerably by showing it is $\NP$-complete to decide whether a bipartite \emph{planar} graph is fall 3-colorable.

If we do not insist on a partition, we prove that deciding whether a triangle-free planar graph contains two disjoint independent dominating sets is $\NP$-complete, strengthening the result of Henning~{et~al.}~\cite{Henning2009} who only showed it for general graphs.

\item In Section~\ref{sec:regular}, we turn our attention to regular graphs. While fall coloring 2-regular graphs is easy, we prove that for every $k \geq 3$, there is some $t$ --- dependant only on $k$ --- such that it is $\NP$-complete to decide whether a $k$-regular graph $G$ is fall $t$-colorable (see the section for precise statements).

\item In Section~\ref{sec:algo}, we conclude by detailing some further algorithmic consequences of our hardness results presented in Section~\ref{sec:planar-bip}.
In addition, we derive moderately exponential-time algorithms for fall coloring.
\end{itemize}

\section{Preliminaries}
For a positive integer $n$, we write $[n] = \{1,2,\ldots,n\}$.
All graphs we consider in this work are undirected and finite.

\paragraph{Graph theory}
Let $G=(V,E)$ be a graph.
For any $s \geq 1$, we denote by $G^s$ the \emph{$s$th power} of $G$, which is $G$ with edges added between every two vertices at a distance no more than~$s$.
In particular, $G^2$ is called the \emph{square} of $G$.
By $G^{\frac{1}{s}}$, we mean $G$ with each of its edges subdivided $s-1$ times.

A \emph{$k$-coloring} of a graph $G$ is a function $c : V \to [k]$. 
A \emph{coloring} is a $k$-coloring for some $k \leq |V|$. 
We say that a coloring $c$ is \emph{proper} if $c(u)\neq c(v)$ for every edge $uv \in E$. 
In particular, if $G$ admits a proper $k$-coloring, we say that $G$ is \emph{$k$-colorable}.
The \emph{chromatic number} of $G$, denoted by $\chi(G)$, is the smallest $k$ such that $G$ is $k$-colorable.

\paragraph{Computational problems} 
The problem of deciding whether a given graph $G$ has $\chi(G) \leq k$ is $\NP$-complete for every $k \geq 3$ (see e.g.,~\cite{Garey1976}).
We refer to this computational problem as \probkCol.
In a closely related problem known as \probChromIndex, the task is to decide whether the \emph{edges} of the input graph can be assigned $k$ colors such that every two adjacent edges receive a distinct color. 
Similarly, for every $k \geq 3$, this problem is also $\NP$-complete even when the input graph is $k$-regular as shown by Leven and Galil~\cite{Leven1983}.

Our focus is on the following problem and its computational complexity.
\begin{framed}
\vspace{-0.2cm}
\noindent \probkFallCol \\
\textbf{Instance:} A graph $G=(V,E)$. \\
\textbf{Question:} Can $V$ be partitioned into $k$ independent dominating sets, i.e., is $k \in \fall(G)$?
\vspace{-0.2cm}
\end{framed}

\section{Hardness results for planar and bipartite graphs}
\label{sec:planar-bip}

In this section, we prove that deciding whether a bipartite planar graph can be fall 3-colored is $\NP$-complete.
Moreover, we show that for every $k \geq 3$, it is $\NP$-complete to decide whether a bipartite graph can be fall $k$-colored.

We begin with the following construction that will be useful to us throughout the section.
\begin{lemma}
\label{lem:3col-to-3fall}
\probThreeCol reduces in polynomial-time to \probThreeFallCol.
\end{lemma}
\begin{proof}
Let $G$ be an instance of \probThreeCol. 
In polynomial time, we will create the following instance $G'$ of \probThreeFallCol, such that $G$ is 3-colorable if and only if $G'$ is fall 3-colorable.

The graph $G'=(V',E')$ is obtained from $G$ by subdividing each edge, and by identifying each vertex in $V$ with a copy of $C_6$.
Formally, we let $$V' = V \cup \{ x_{uv} \mid uv \in E \} \cup \{ w_{vi} \mid v \in V, i \in [5] \}, \text{ and}$$ $$E' = \{ ux_{uv}, vx_{uv} \mid uv \in E \} \cup \{ vw_{v1}, vw_{v5}, w_{vi}w_{vi+1} \mid v \in V, i \in [4] \}.$$
This finishes the construction of $G'$.

Let $c : V \to [3]$ be a proper vertex-coloring of $G$, and let us construct a fall 3-coloring $c' : V' \to [3]$ as follows.
We retain the coloring of the vertices in $V$, that is, $c'(v) = c(v)$ for every $v \in V$.
Then, as the degree of each $x_{uv}$ is two, it holds that in any valid fall 3-coloring of $G'$, the colors from $[3]$ must be bijectively mapped to the closed neighborhood $\{ u,v,x_{uv}\}$ of $v_{uv}$.
Thus, we set $c'(x_{uv}) = f$, where $f$ is the unique color in $[3]$ neither $c(u)$ nor $c(v)$.
Finally, consider an arbitrary vertex $v \in V$. 
Without loss of generality, suppose $c(v) = 1$. 
We will then finish the vertex-coloring $c'$ as follows (see Figure~\ref{fig:fk345}, where $C_6 \simeq F_3$):
\[ 
 c'(w_{v3}) = 1\,,\;
 c'(w_{v1}) = 2\,,\;
 c'(w_{v4}) = 2\,,\;
 c'(w_{v2}) = 3\,,\;
 c'(w_{v5}) = 3\,.
\]
It is straightforward to verify that $c'$ is indeed a fall 3-coloring of $G'$.

For the other direction, let $c'$ be a fall 3-coloring of $G'$.
Again, because the degree of each $x_{uv}$ is two, it holds that $c'(u) \neq c'(v)$.
Therefore, $c'$ restricted to $G$ is a proper 3-coloring for $G$.
This concludes the proof.
\end{proof}
\noindent Combining the previous lemma with the well-known fact that deciding whether a planar graph of maximum degree~4 can be properly 3-colored is $\NP$-complete~\cite{Garey1976}, we obtain the following.
\begin{corollary}
\label{cor:bip-planar-npc}
It is $\NP$-complete to decide whether a bipartite planar graph $G$ of maximum degree~6 is fall 3-colorable.
\end{corollary}
\begin{proof}
It suffices to observe that the construction of Lemma~\ref{lem:3col-to-3fall} does not break planarity (i.e., if $G$ is planar, so is $G'$) and that after subdividing the edges of $G$ the resulting graph $G'$ is bipartite.
Finally, a vertex $v$ of degree $\Delta \leq 4$ in $G$ has degree $\Delta + 2 \leq 6$ in $G'$ after $v$ is identified with a copy of $C_6$, whereas the new vertices (in copies of $C_6$ or from subdividing) have degree~$2$.
\end{proof}

In order to show that fall $k$-coloring is hard for every $k \geq 3$ for the class of bipartite graphs, we make use of the following construction. As a reminder, $G \times H$ is the \emph{categorical product} of graphs $G$ and $H$ with $V(G \times H) = V(G) \times V(H)$ and $(u_1, v_1)(u_2, v_2) \in E(G \times H)$ when $u_1u_2 \in E(G)$ and $v_1v_2 \in E(H)$.
\begin{proposition}
\label{prop:fk-fallkcol}
For every $k \geq 3$, the graph $F_k = K_2 \times K_k$ is bipartite and uniquely fall $k$-colorable.
\end{proposition}
\begin{proof}
It is well-known that $G \times H$ is bipartite if either $G$ or $H$ is bipartite.
Thus, as $K_2$ is bipartite, so is $F_k$.

It follows from  Dunbar~{et al.}~\cite[Theorem~6]{Dunbar2000} that if $s$ and $k$ are distinct positive integers both greater than one, then $K_s \times K_k$ has a fall $k$-coloring. 
In our case, $s = 2$ and $k \geq 3$, so $F_k$ admits a fall $k$-coloring.
The fact that $F_k$ has a unique fall $k$-coloring follows from~\cite[Theorem~15]{Mitillos2016}.
\end{proof}

\begin{figure}[t]
    \centering
    \begin{minipage}{0.33\textwidth}
        \centering
        \includegraphics[width=0.9\textwidth]{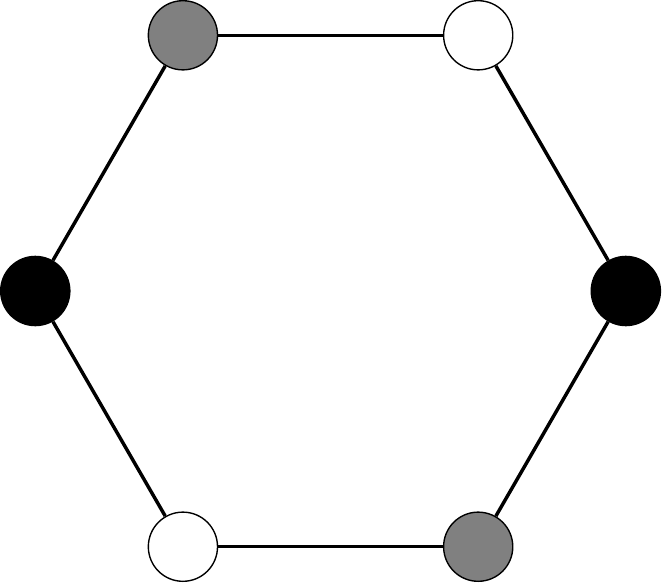} 
    \end{minipage}\hfill
    \begin{minipage}{0.33\textwidth}
        \centering
        \includegraphics[width=0.9\textwidth]{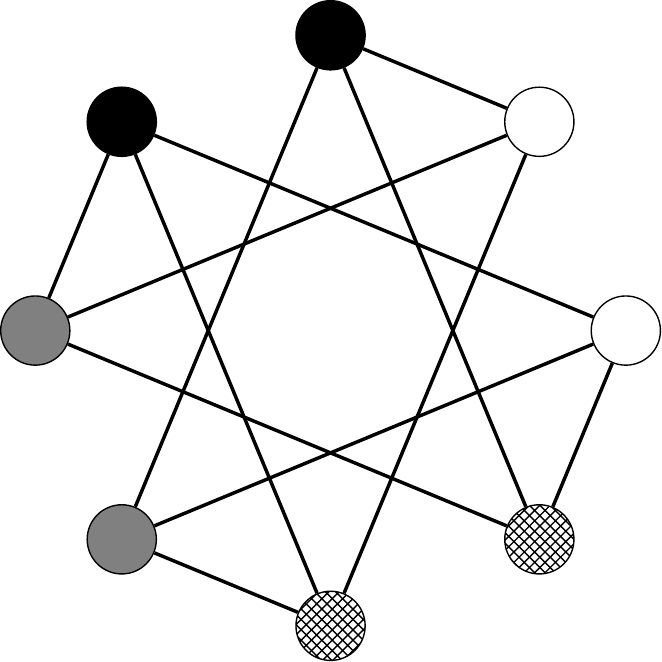} 
    \end{minipage}\hfill
    \begin{minipage}{0.33\textwidth}
        \centering
        \includegraphics[width=0.9\textwidth]{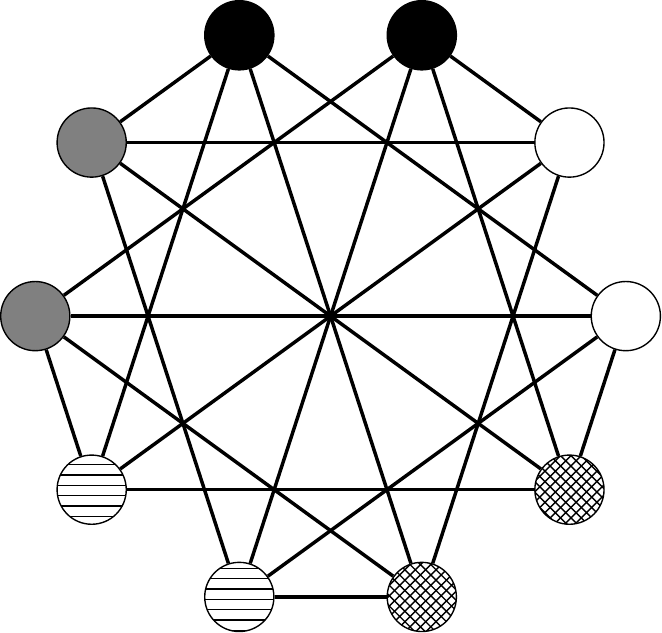}     
    \end{minipage}\hfill
    \caption{The graphs $F_k = K_2 \times K_k$ for $3 \leq k \leq 5$ each with a fall $k$-coloring shown.}
    \label{fig:fk345}
    \vspace{-0.5cm}
\end{figure}

\noindent We are then ready to proceed with the reduction, following the idea of Lemma~\ref{lem:3col-to-3fall}.
\begin{lemma}
\label{lem:kcol-to-kfall}
For every $k \geq 4$, it is $\NP$-complete to decide whether a bipartite graph $G$ is fall $k$-colorable.
\end{lemma}
\begin{proof}
We show this by extending the method in Lemma~\ref{lem:3col-to-3fall}. Given a graph $G = (V, E)$, we construct in polynomial time a bipartite graph $G' = (V', E')$, so that $G'$ is fall $k$-colorable if and only if $G$ is $k$-colorable. Then, since it is $\NP$-complete to decide whether $G$ is $k$-colorable, the result will follow.

As before, we begin by subdividing every edge of $G$ once, and identifying each vertex of $V$ with a copy of $F_k$. Then, for each vertex $x_{uv}$ created by subdividing some edge $uv$ of $G$, we create $k - 3$ disjoint copies of $F_k$, and arbitrarily select one vertex in each such copy to make adjacent to $x_{uv}$. Note that when $k = 3$, this simplifies to the construction in Lemma~\ref{lem:3col-to-3fall}.

First, we observe that the resulting graph $G'$ is bipartite. It consists of one copy of $G^{\frac{1}{2}}$ and multiple disjoint copies of $F_k = K_2 \times K_k$, connected to $G^{\frac{1}{2}}$ by either cut-vertices or cut-edges. Since $G^{\frac{1}{2}}$ and $F_k$ are both bipartite, $G'$ is bipartite as well.

Let $c$ be a proper $k$-coloring of $G$. We extend it to a fall $k$-coloring $c'$ of $G'$ as follows. For every edge $uv \in E$, the vertex $x_{uv}$ is colored arbitrarily with some color distinct from both $c(u)$ and $c(v)$. Then, its remaining $k - 3$ neighbors are each given a different color, so that $x_{uv}$ is colorful. Now every copy of $F_k$ in the graph has exactly one colored vertex; since $F_k$ has a unique fall $k$-coloring (up to isomorphism) by Proposition~\ref{prop:fk-fallkcol}, we use this to complete $c'$.

For the other direction, let $c'$ be a fall $k$-coloring of $G'$. Then, since each $x_{uv}$ has $k - 1$ neighbours and is colorful, $c'(u) \neq c'(v)$. Restricting $c'$ to $V$, we obtain a proper $k$-coloring of $G$.
\end{proof}

\begin{theorem}
For every $k \geq 3$, it is $\NP$-complete to decide whether a bipartite graph $G$ is fall $k$-colorable.
\end{theorem}
\begin{proof}
The proof follows by combining Lemma~\ref{lem:3col-to-3fall} with Lemma~\ref{lem:kcol-to-kfall}.
\end{proof}
\noindent We also observe the following slightly stronger corollary.
\begin{corollary}
For every $k \geq 3$, it is $\NP$-complete to decide whether a bipartite graph $G$ of maximum degree $3(k-1)$ is fall $k$-colorable.
\end{corollary}
\begin{proof}
We use Lemma~\ref{lem:3col-to-3fall} and Lemma~\ref{lem:kcol-to-kfall} with the fact that deciding whether a graph $G$ has $\chi(G) \leq k$ is $\NP$-complete for every $k \geq 3$ even when $G$ has maximum degree $\Delta = 2k-2$ (see~\cite[Theorem~3]{Maffray1996}).
Now, $F_k = K_2 \times K_k$ is $(k-1)$-regular, so a vertex of degree $\Delta$ in $G$ has degree $\Delta + k - 1 = 3k-3$ in $G'$. At the same time, the new vertices (from subdividing, or copies of $F_k$) have degree at most $k < 3k - 3$.
The claim follows.
\end{proof}

After Corollary~\ref{cor:bip-planar-npc}, it is natural to wonder what are the weakest additional constraints to place on the structure of a planar graph so that say fall 3-coloring is solvable in polynomial time.
In the following, we show that maximal outerplanar graphs with at least three vertices admit a fall 3-coloring, and in fact no other fall colorings.
We begin with the following two propositions; for short proofs of both we refer the reader to~\cite{Mitillos2016}.
\begin{proposition}
\label{prop:chordal-fall}
Let $G$ be a chordal graph. Then either $\fall(G) = \emptyset$ or $\fall(G) = \{ \delta(G) + 1 \}$.
\end{proposition}
\begin{proposition}
\label{prop:uniq-col-fall}
If $G$ is a uniquely $k$-colorable graph, then $G$ is fall $k$-colorable.
\end{proposition}

\noindent These results will be combined with the following theorem.

\begin{theorem}[Chartrand and Geller~\cite{Chartrand1969}]
\label{thm:outerplanar-unique3}
An outerplanar graph $G$ with at least three vertices is uniquely 3-colorable if and only if it is maximal outerplanar.
\end{theorem}

\noindent The claimed result is now obtained as follows.

\begin{theorem}
Let $G$ be a maximal outerplanar graph with at least three vertices. Then $\fall(G) = \{ 3 \}$.
\end{theorem}
\begin{proof}
As every maximal outerplanar graph $G$ is chordal, it follows by Proposition~\ref{prop:chordal-fall} that $\fall(G) = \emptyset$ or $\fall(G) = \{ \delta(G) + 1 \}$.
It is well-known that every maximal outerplanar graph has at least two vertices of degree two.
By combining Theorem~\ref{thm:outerplanar-unique3} with Proposition~\ref{prop:uniq-col-fall}, we have that $\fall(G) = \{ \delta(G) + 1 \} = \{ 3 \}$.
\end{proof}

Also, in the light of Corollary~\ref{cor:bip-planar-npc}, it should be recalled that any proper 2-coloring of a connected bipartite graph is a fall 2-coloring.
Thus, the statement of Corollary~\ref{cor:bip-planar-npc} would not hold for the case of $k = 2$ colors (unless $\P = \NP$).
However, what if we do not insist on a partition of the vertices but are merely interested in the existence of two disjoint independent dominating sets?
As we will show, this problem is $\NP$-complete for planar graphs; in fact even those that are triangle-free.
This result is a considerable strengthening of an earlier result of Henning~{et~al.}~\cite{Henning2009}, who showed it only for general graphs.

\begin{figure}[t]
    \centering
        \includegraphics[scale=0.75,keepaspectratio]{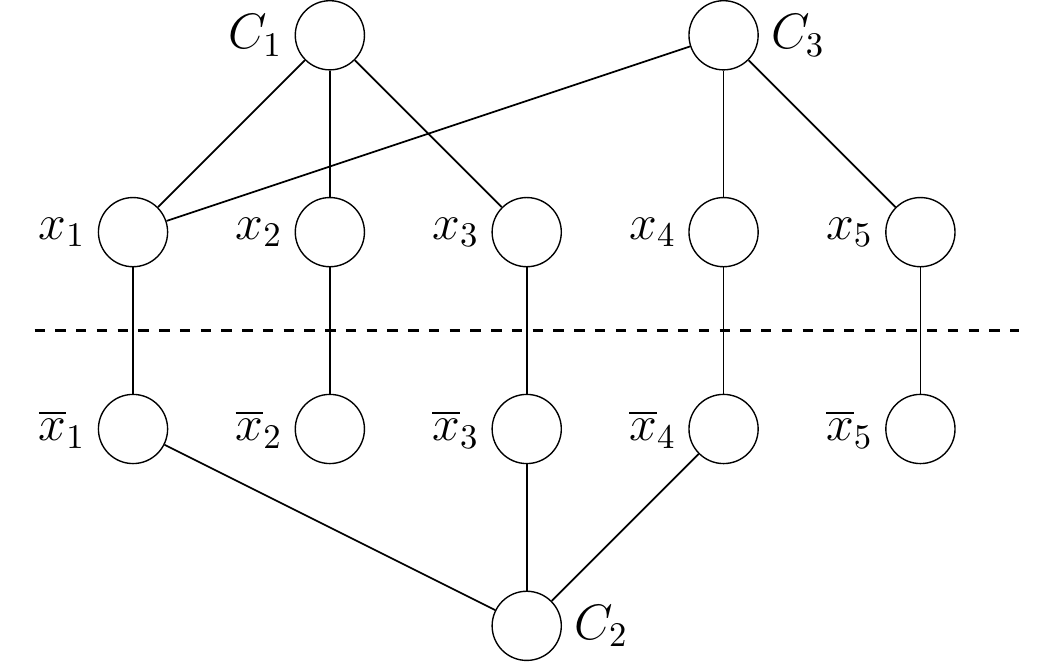} 
    \caption{An instance $\varphi = (x_1 \vee x_2 \vee x_3) \wedge (\overline{x}_1 \vee \overline{x}_3 \vee \overline{x}_4) \wedge (x_1 \vee x_4 \vee x_5)$ of \probPlanarMonotoneThreeSat, which always admits a planar drawing $G(\varphi)$. Conceptually, the dashed horizontal line separates the upper part containing all positive literals and clauses from the lower part containing all negative literals and clauses.}
    \label{fig:monotone-sat}
    \vspace{-0.5cm}
\end{figure}

In the \probPlanarMonotoneThreeSat problem, we are given a 3-SAT formula $\varphi$ with $m$ clauses over $n$ variables $x_1$, $x_2, \ldots x_n$, where each clause $c_1$, $c_2, \ldots, c_m$ comprises either three positive literals or three negative literals.
We call such clauses \emph{positive} and \emph{negative}, respectively.
Moreover, the associated graph $G(\varphi)$ has a 2-clique (i.e., an edge) $\{ x_i, \overline{x}_i \}$ for each variable $x_i$, a vertex for each $c_j$, and an edge between a literal contained in a clause and the corresponding clause.
In particular, $G(\varphi)$ admits a planar drawing such that every 2-clique sits on a horizontal line with the line intersecting their edges.
In addition, every positive clause is placed above the line, while every negative clause is placed below the line (see Figure~\ref{fig:monotone-sat}).
The fact that \probPlanarMonotoneThreeSat is $\NP$-complete and that $G(\varphi)$ admits the claimed planar drawing follows from {de Berg and Khosravi}~\cite{DeBerg2012}.

\begin{theorem}
It is $\NP$-complete to decide whether a given triangle-free planar graph has two disjoint independent dominating sets.
\end{theorem}
\begin{proof}
The proof is by a polynomial-time reduction from \probPlanarMonotoneThreeSat, whose input is a monotone 3-SAT instance $\varphi$ with a set of $m$ clauses $\mathcal{C} = \{ C_1, C_2, \ldots, C_m \}$ over the $n$ variables $\mathcal{X} = \{ x_1, x_2, \ldots, x_n \}$.
Since our goal is to construct a graph $G'$ that is both triangle-free and planar, it is convenient to start from a planar drawing of $G(\varphi)$ as described, and proceed as follows.

\begin{figure}[t]
    \centering
        \includegraphics[width=0.3\textwidth,keepaspectratio]{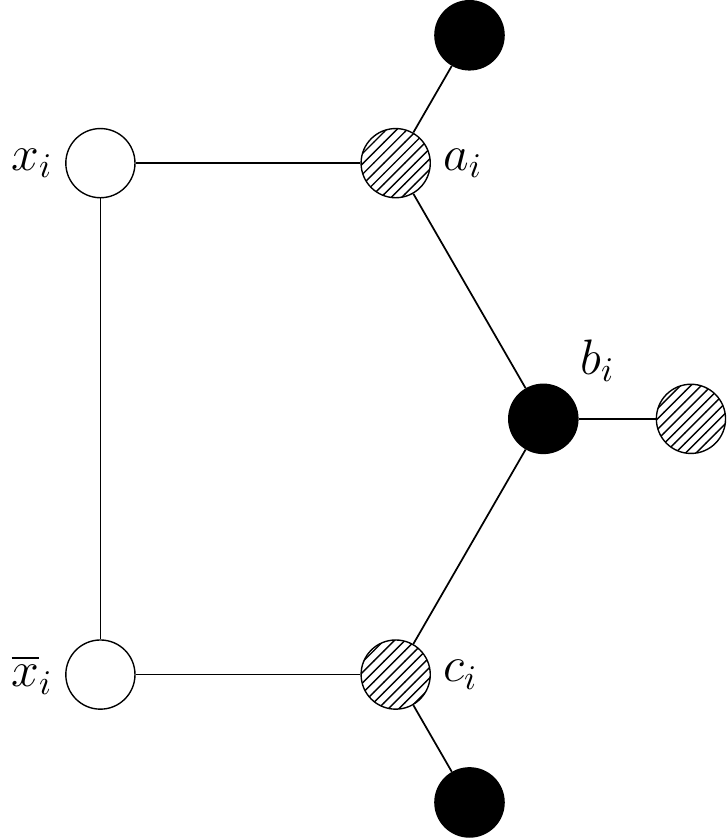} 
    \caption{The variable gadget $X_i$. If $I_1$ and $I_2$ are two disjoint independent dominating sets, then $a_i$ and $c_i$ must both be in either $I_1$ or $I_2$. Furthermore, exactly one of $x_i$ and $\overline{x}_i$ can be in $I_1 \cup I_2$.}
    \label{fig:vargadget}
    \vspace{-0.5cm}
\end{figure}

For each variable $x_i$, we extend its corresponding 2-clique in $G(\varphi)$ by replacing it with the following variable gadget $X_i$ (see Figure~\ref{fig:vargadget}).
Here, $X_i$ is a 5-cycle on the vertices $x_i$, $a_i$, $b_i$, $c_i$, and $\overline{x}_i$ (in clockwise order) with a pendant vertex attached to each of $a_i$, $b_i$, and $c_i$.
Otherwise, we retain the structure of $G(\varphi)$ finishing our construction of $G'$.
Clearly, as $G(\varphi)$ is planar and triangle-free, so is $G'$. 
We will then prove that $\varphi$ is satisfiable if and only if $G'$ contains two disjoint independent dominating sets.

Let $\varphi$ be satisfiable under the truth assignment $\tau = \{0,1\}^n$.
We construct two disjoint independent dominating sets $I_1$ and $I_2$ as follows.
For each $i \in [n]$, if $\tau$ sets $x_i$ to $1$, put $x_i$ to $I_1$.
Otherwise, if $\tau$ sets $x_i$ to $0$, put $\overline{x}_i$ to $I_1$.
Put every $b_i$ to $I_1$, and every $a_i$ and $c_i$ to $I_2$.
The pendant vertices of $a_i$ and $c_i$ are put to $I_1$, while the pendant vertices of $b_i$ are put to $I_2$.
For each $j \in [m]$, put $C_j$ in $I_2$.
Observe that both $I_1$ and $I_2$ are independent.
Moreover, every vertex of $X_i$ is dominated by a vertex in $I_1$, and also by a vertex in $I_2$.
Every vertex $C_j$ is dominated by a vertex in $I_2$, and since $\tau$ is a satisfying assignment, $C_j$ must also be adjacent to a vertex in $I_1$.
We conclude that $I_1$ and $I_2$ are disjoint independent dominating sets of $G'$. 

Conversely, suppose that $I_1$ and $I_2$ are two disjoint independent dominating sets of $G'$.
Clearly, each clause $C_j$ for $j \in [m]$ must be dominated by at least one $x_i$ (or $\overline{x}_i$ in the case of a negative clause).
For each $i \in [n]$, observe that $a_i$ and $c_i$ must both be in $I_1$ or $I_2$ (for otherwise the pendant of $b_i$ could not be dominated by both a vertex of $I_1$ and a vertex of $I_2$). 
It follows that at most one of $x_i$ and $\overline{x}_i$ can be in $I_1 \cup I_2$.
Thus, the vertices in $I_1 \cup I_2$ corresponding to variable vertices encode a satisfying assignment $\tau$ for $\varphi$. 
Finally, notice that neither $x_i$ or $\overline{x}_i$ are in $I_1 \cup I_2$, the truth value of the corresponding variable does not affect the satisfiability of $\varphi$, and can thus be set arbitrarily in $\tau$.
\end{proof}

\section{Hardness results for regular graphs}
\label{sec:regular}

In this section, we consider the complexity of fall coloring regular graphs.
For connected 2-regular graphs (i.e., cycles), it is not difficult to verify that $2 \in \fall(C_n)$ if and only if $2 \mid n$ and that $3 \in \fall(C_n)$ if and only if $3 \mid n$ with no other integer being in $\fall(C_n)$, for any $n$ (see e.g.,~\cite{Dunbar2000,Mitillos2016}).
However, as we will show next, the problem of fall coloring 3-regular graphs is considerably more difficult.

We begin by recalling the following result.
\begin{theorem}[Heggernes and Telle~\cite{Heggernes1998}]
It is $\NP$-complete to decide if the square of a cubic graph is 4-chromatic.
\end{theorem}
\noindent In addition, we make use of the following fact.
\begin{theorem}[\cite{Mitillos2016}]
A $k$-regular graph $G$ is fall $(k+1)$-colorable if and only if $G^2$ is $(k+1)$-chromatic.
\end{theorem}
\noindent By combining the two previous theorems, we arrive at the following.
\begin{theorem}
It is $\NP$-complete to decide whether a 3-regular graph $G$ is fall 4-colorable.
\end{theorem}
\noindent The previous result suggests that there may be similar intractable fall-colorability problems for regular graphs of higher degree. With this in mind, we use different constructions for regular graphs, to show that fall coloring $k$-regular graphs for $k > 3$ is $\NP$-complete as well.
\begin{theorem}
\label{thm:reg2k2}
For every $k \geq 3$, it is $\NP$-complete to decide whether a $(2k-2)$-regular graph $G$ is fall $k$-colorable.
\end{theorem}
\begin{proof}
The proof is by a polynomial-time reduction from \probChromIndex, where we assume that $k \geq 3$ and that the input graph $G$ is $k$-regular.
Let $G' = L(G)$, that is, $G'$ is the line graph of $G$.
Because $G$ is $k$-regular, it is straightforward to verify that $G'$ is $(2k-2)$-regular.
We then prove that $G$ admits a proper edge $k$-coloring if and only if $G'$ admits a fall $k$-coloring.

Let $h$ be a proper edge $k$-coloring of $G$.
We construct a vertex-coloring $c'$ of $G'$ as follows.
Let $c'(x_{uv}) = h(uv)$, where $uv \in E(G)$ and $x_{uv}$ is the vertex of $G'$ corresponding to the edge $uv$.
By construction, $x_{uv}$ for every $uv \in E(G)$ is adjacent to precisely the vertices corresponding to the edges adjacent to $u$ and $v$.
Since $h$ is a proper edge-coloring, $h$ has colored these edges differently from $h(uv) = c'(x_{uv})$. Furthermore, all the $k$ edges incident to the same vertex will receive $k$ different colors. As such, the corresponding vertices in $G'$ will all be colorful.
We conclude that $c'$ is a fall $k$-coloring for $G'$.

In the other direction, let $c'$ be a fall $k$-coloring of $G'$.
Consider any $x_{uv}$ of $G'$.
Because $c'$ is a fall $k$-coloring, each neighbor of $x_{uv}$ has received a distinct color.
Again, $x_{uv}$ is adjacent to precisely the vertices that correspond to edges adjacent to $u$ and $v$ in $G$.
Thus, we obtain immediately a proper edge $k$-coloring $h$ from $c'$, concluding the proof.
\end{proof}

We can get a similar result for regular graphs with vertices of odd degree by using the \emph{Cartesian product} of graphs $G$ and $H$, denoted as $G \Box H$. As a reminder, $V(G \Box H) = V(G) \times V(H)$ and $(u_1,v_1)(u_2,v_2) \in E(G \Box H)$ when either $u_1u_2 \in E(G)$ and $v_1 = v_2$ or $u_1 = u_2$ and $v_1v_2 \in E(H)$.

\begin{theorem}
For every $k \geq 3$, it is $\NP$-complete to decide whether a $(2k-1)$-regular graph $G$ is fall $k$-colorable.
\end{theorem}
\begin{proof}
It suffices to modify the graph $G'$ from the proof of Theorem~\ref{thm:reg2k2} to obtain a graph with mostly the same structure; in particular, a graph which can be fall $k$-colored exactly when $G'$ can, but whose vertices have common degree on more than those of $G'$.
One such construction is $G'' = G' \Box K_2$. It is easy to see that $\fall(G'') = \fall(G')$ (see~\cite{Mitillos2016}), so $G''$ has exactly the properties we require. 
\end{proof}

\section{Further algorithmic consequences}
\label{sec:algo}
In this section, we give further algorithmic consequences of our hardness results.

A popular measure --- especially from an algorithmic viewpoint --- for the ``tree-likeness'' of a graph is captured by the notion of treewidth.
Here, a \emph{tree decomposition} of $G$ is a pair $(T,\{X_i : i\in I\})$ where $X_i \subseteq V$, $i\in I$, and $T$ is a tree with elements of $I$ as nodes such that:
\begin{enumerate}
\item for each edge $uv \in E$, there is an $i\in I$ such that $\{u,v\} 
\subseteq X_i$, and
\item for each vertex $v \in V$, $T[\SB i\in I \SM v\in X_i \SE]$ is a tree with at least one node.
\end{enumerate}
The \emph{width} of a tree decomposition is $\max_{i \in I} |X_i|-1$.
The \emph{treewidth} of $G$, denoted by $\tw(G)$, is the minimum width taken over all tree decompositions of $G$. 

The following result is easy to observe, but we include its proof for completeness.
\begin{theorem}
\label{thm:twpoly}
Let $G$ be a graph of bounded treewidth. The fall set $\fall(G)$ can be determined in polynomial time.
\end{theorem}
\begin{proof}
It is well-known that every graph of treewidth at most $p$ has a vertex of degree at most $p$.
It follows that the largest integer in $\fall(G)$ is $\psi_{fall}(G) \leq p + 1$.
Thus, it suffices to test whether $i \in \fall(G)$ for $i \in \{1,2,\ldots,p+1\}$.
Furthermore, the fall $i$-colorability of $G$ can be tested in polynomial time by the result of Telle and Proskurowski~\cite{Telle1997}.
(Alternatively, this can be seen by observing that fall $i$-colorability can be characterized in monadic second order logic, and then applying the result of Courcelle~\cite{Courcelle1990}).
The claim follows.
\end{proof}
At this point, it will be useful to recall that  a \emph{parameterized problem} $I$ is a pair $(x,k)$, where $x$ is drawn from a fixed, finite alphabet and $k$ is an integer called the \emph{parameter}. 
Then, a \emph{kernel} for $(x,k)$ is a polynomial-time algorithm that returns an instance $(x',k')$ of $I$ such that $(x,k)$ is a YES-instance if and only if $(x',k')$ is a YES-instance, and $|x'| \leq g(k)$, for some computable function $g : \mathbb{N} \to \mathbb{N}$.
If $g(k)$ is a polynomial (exponential) function of $k$, we say that $I$ admits a polynomial (exponential) kernel (for more, see Cygan~{et~al.}~\cite{fpt-book}).

A consequence of Theorem~\ref{thm:twpoly} is that for every $k \geq 1$, \probkFallCol admits an exponential kernel.
Here, we will observe that Lemma~\ref{lem:3col-to-3fall} actually proves that this is the best possible, i.e., that there is no polynomial kernel under reasonable complexity-theoretic assumptions.

First, the gadget $C_6$ each vertex of $G$ is identified with in Lemma~\ref{lem:3col-to-3fall} has treewidth two.
Second, this identification increases the treewidth of $G$ by only an additive constant.
To make use of these facts, we recall that Bodlaender~{et al.}~\cite{Bodlaender2009} proved that \probThreeCol does not admit a polynomial kernel parameterized by treewidth unless \containment.
At this point, it is clear that the proof of Lemma~\ref{lem:3col-to-3fall} is actually a parameter-preserving transformation (see~\cite[Theorem~15.15]{fpt-book} or~\cite[Section~3]{Bodlaender2011}) guaranteeing $\tw(G') \leq \tw(G) + 2$.
We obtain the following.
\begin{theorem}
\probThreeFallCol parameterized by treewidth does not admit a polynomial kernel unless \containment.
\end{theorem}

A further consequence of Lemma~\ref{lem:3col-to-3fall} is that fall $k$-coloring is difficult algorithmically, even when the number of colors is small and the graph is planar.
To make this more precise, we recall the well-known \emph{exponential time hypothesis} (ETH), which is a conjecture stating that there is a constant $c > 0$ such that \probThreeSAT cannot be solved in time $O(2^{cn})$, where $n$ is the number of variables.
\begin{corollary}
\label{cor:planar-fall-eth}
\probThreeFallCol for planar graphs cannot be solved in time $2^{o(\sqrt{n})}$ unless ETH fails, where $n$ is the number of vertices. However, the problem admits an algorithm running in time $2^{O(\sqrt{n})}$ for planar graphs.
\end{corollary}
\begin{proof}
It suffices to observe that the graph $G'$ obtained in the proof of Lemma~\ref{lem:3col-to-3fall} has size linear in the size of the input graph $G$.
The claimed lower bound then follows by a known chain of reductions originating from \probThreeSAT (see e.g.,~\cite[Theorem~14.3]{fpt-book}).

The claimed upper bound follows from combining the single-exponential dynamic programming algorithm on a tree decomposition of van Rooij~{et al.}~\cite{vanRooij2009} with the fact that an $n$-vertex planar graph has treewidth $O(\sqrt{n})$ (for a proof, see~\cite[Theorem~3.17]{Fomin2006}).
\end{proof}

Finally, observe that the naive exponential-time algorithm for deciding whether $k \in \fall(G)$ enumerates all possible $k$-colorings of $V(G)$ and thus requires $k^n n^{O(1)}$ time. A much faster exponential-time algorithm is obtained as follows.
\begin{theorem}
\probkFallCol can be solved in $3^n n^{O(1)}$ time and polynomial space. In exponential space, the time can be improved to $2^n n^{O(1)}$.
\end{theorem}
\begin{proof}
The claimed algorithms are obtained by reducing the problem to \probSetPart, in which we are given a universe $U = [n]$, a set family $\mathcal{F} \subseteq 2^U$, and an integer $k$.
The goal is to decide whether $U$ admits a partition into $k$ members.

We enumerate all the $2^n$ vertex subsets of the $n$-vertex input graph $G$ and add precisely those to $\mathcal{F}$ that form an independent dominating set, a property decidable in polynomial time.
To finish the proof, we apply the result of Bj\"{o}rklund~et al.~\cite[Thms.~2~and~5]{Koivisto2009} stating that \probSetPart can be solved in $2^n n^{O(1)}$ time. 
Further, if membership in $\mathcal{F}$ can be decided in $n^{O(1)}$ time, then 
\probSetPart can be solved in $3^n n^{O(1)}$ time and $n^{O(1)}$ space.
\end{proof}

\section{Conclusions}
We further studied the problem of partitioning a graph into independent dominating sets, also known as fall coloring.
Despite the centrality of the concepts involved, independence and domination, a complete understanding of the complexity fall coloring is lacking.
Towards this end, our work gives new results and strengthens previously known hardness results on structured graph classes, including various planar graphs, bipartite graphs, and regular graphs.

An interesting direction for future work is finding combinatorial algorithms for fall coloring classes of bounded treewidth (or in fact, bounded cliquewidth).
Indeed, the algorithms following from the proof of Theorem~\ref{thm:twpoly} are not practical.
For concreteness, one could consider outerplanar graphs or cographs.

\bibliographystyle{splncs04}
\bibliography{bibliography}
\end{document}